\documentclass[sigconf,authorversion,noacm]{acmart}
\usepackage[activeacute,english]{babel} 
\usepackage[utf8]{inputenc}
\usepackage{verbatim} 
\usepackage{graphicx} 
\usepackage{enumerate,mdwlist}  
\usepackage{multirow} 
\usepackage{amsmath,amsfonts,amsthm} 
\usepackage{bm} 
\usepackage{empheq} 
\usepackage{bbm,dsfont} 
\usepackage{mathrsfs} 
\numberwithin{equation}{section} 
\usepackage{float} 
\usepackage{empheq} 
\usepackage{balance}

\usepackage{listings}
\theoremstyle{plain}
\newtheorem{lem}{Lemma}[section]
\newtheorem{remark}[lem]{Remark}

\newcommand{\CC}{\mathbb{C}}
\newcommand{\ZZ}{\mathbb{Z}}
\newcommand{\QQ}{\mathbb{Q}}
\newcommand{\NN}{\mathbb{N}}

\newcommand{\OB}{\mathcal{O}_B}
\newcommand{\OO}{\mathcal{O}}
\newcommand{\sOB}{\widetilde{\mathcal{O}}_B}
\newcommand{\sOO}{\widetilde{\mathcal{O}}}

\newcommand{\mah}{\mathcal{M}}

\newcommand{\gdisc}{\mathrm{GDisc}}
\newcommand{\lgdisc}{\mathrm{lGDisc}}
\newcommand{\sepa}{\mathrm{sep}}
\newcommand{\lsepa}{\mathrm{lsep}}
\newcommand{\sres}{\mathrm{sres}}
\newcommand{\mult}[1]{\mu_{_{#1}}}
\newcommand{\I}{\mathcal{I}}
\newcommand{\J}{\mathcal{J}}

\newcommand{\X}{\bm{X}}
\newcommand{\x}{\bm{x}}
\newcommand{\Y}{\bm{Y}}
\newcommand{\y}{\bm{y}}
\renewcommand{\a}{\bm{a}}

\newcommand{\res}{\mathrm{res}}

\def\norm#1{\| #1 \|}

\newcommand{\lc}{\text{lc}}

\makeatletter
\let\original@algocf@latexcaption\algocf@latexcaption
\long\def\algocf@latexcaption#1[#2]{%
  \@ifundefined{NR@gettitle}{%
    \def\@currentlabelname{#2}%
  }{%
    \NR@gettitle{#2}%
  }%
  \original@algocf@latexcaption{#1}[{#2}]%
}
\makeatother

\setcopyright{none}
\renewcommand\footnotetextcopyrightpermission[1]{} 
\begin{document}
\fancyhead{}

\pagestyle{plain}

\title[On Isolating Roots in a Multiple Field Extension]{On Isolating Roots in a Multiple  Field Extension}
\author{{Christina Katsamaki}}
\email{christina.katsamaki@inria.fr}
\affiliation{%
  \institution{INRIA Paris, \\ Sorbonne Universit\'e and Paris Universit\'e}
  \streetaddress{4 place Jussieu}
  \city{F-75005, Paris}
  \country{France}
}

\author{Fabrice Rouillier}
\email{Fabrice.Rouillier@inria.fr}
\affiliation{%
  \institution{INRIA Paris, \\ Sorbonne Universit\'e and Paris Universit\'e}
  \streetaddress{4 place Jussieu}
  \city{F-75005, Paris}
  \country{France}
}

\renewcommand{\authors}{Christina Katsamaki and Fabrice Rouillier}
\begin{abstract}
We address univariate root isolation when the polynomial's coefficients are in a multiple field extension. We consider a polynomial $F \in L[Y]$, where $L$ is a multiple algebraic extension of $\QQ$. 
We provide aggregate bounds for $F$ and algorithmic and bit-complexity results for the problem of isolating its roots. 

For the latter problem we follow a common approach based on univariate root isolation algorithms. For the particular case where $F$ does not have multiple roots, we achieve a bit-complexity in $\sOB(n d^{2n+2}(d+n\tau))$, where $d$ is the total degree and $\tau$ is the bitsize of the involved polynomials.
In the general case we need to enhance our algorithm with a preprocessing step that determines the number of distinct roots of $F$. We follow a numerical, yet certified, approach that has bit-complexity
$\sOB(n^2d^{3n+3}\tau + n^3 d^{2n+4}\tau)$.
\end{abstract}



\maketitle

\section{Introduction}

We consider the problem of isolating the (complex) roots of a univariate polynomial over a multiple algebraic field extension --the coefficients of the polynomial are multivariate polynomial functions evaluated at algebraic numbers--. Solving in a field extension is a common problem in computational mathematics; for example it arises in the topology computation of plane curves  \cite{Diatta2022,ks-jsc-2012} or it can be seen as a sub-problem in the resolution of triangular systems \cite{CHEN2012610,Xia2002AnAF} and regular chains \cite{Boulier10}.

For $n\ge2$, we consider $F_1\in \ZZ[X_1], \dots, F_n \in \ZZ[X_n]$ univariate polynomials of degree at most $M$ and bitsize $\Lambda$ and $F \in \ZZ[X_1,\dots, X_n,$
$Y]$ of total degree at most $d$ and bitsize $\tau$.
We want to isolate the roots of the system
\begin{equation}
\begin{aligned}
\label{eq:introsystem}
 \begin{Bmatrix} 
 F_1(X_1)=0\, , \dots,  F_n(X_n)=0\,,     \\
F(X_1,\dots,X_n,Y)=0\, .
\end{Bmatrix} 
\end{aligned}
\end{equation}
In theory, we can solve the system as follows: first, we isolate the roots of all the univariate polynomials $F_1, \dots, F_n$. Then, for every root $\x = (x_1, \dots, x_n) \in \CC^n $ of $\{F_1 = \dots = F_n= 0\}$ we employ Pan's algorithm \cite{PAN2002701} for the approximate factorization of the univariate polynomial $F(\x,Y)$, with the worst case precision; the
approximate factorization algorithm returns as many root approximations as the degree of $F(\x,Y)$ in $Y$ is, and by utilizing the worst case precision, we can identify the root approximations that correspond to the same root. This method leads to good worst-case bit-complexity estimates (Rem.~\ref{rem:pan}). Nevertheless, this is at the price of always requiring  to perform computations using the maximum precision and it cannot lead to a practical algorithm. 
For example, if $n=2$, $d=10$, and $\tau\in \OO(1)$, then 
we \emph{} have to work with $> 10^4$ bits in all of our computations.
Our goal is to introduce an adaptive algorithm, depending on the multiplicities of the roots and on their pairwise distances, so we will follow a different approach.


Isolating the roots of the system in Eq.~(\ref{eq:introsystem}) has not been treated in the literature currently, but only in a simplified setting (e.g. \cite{Rump77,Johnson97}).  In \cite{StrETmult} they consider the same problem  when $F$ does not have multiple roots; it is a generalization  of a prior work for a simple algebraic extension \cite{STRZEBONSKI201931,Johnson97}.  They propose three methods. The first one computes the minimal polynomial of the system and uses multivariate resultants. The second one is based on Sturm's algorithm and the third one on solving directly the polynomial using univariate root isolation algorithms, similarly to ours. The last method is the most efficient with a bit-complexity of $\sOB(n^3N^{2n+3})$, where $N$ is a bound on the size of the input polynomials,  without any assumptions on the input.  
In \cite{Diatta2022}, it is the first time the general problem, in the simple extension setting, is addressed in literature \cite{ks-jsc-2012,STRZEBONSKI201931,MEHLHORN201534,KOBEL2015206}. The authors provide precise amortized separation bounds for the polynomial and complexity bounds for isolating the roots. 
We provide further details on their method in the sequel, since we share many ideas.
We could also solve the system in Eq.~(\ref{eq:introsystem}) by applying a general algorithm for zero-dimensional square systems of an expected complexity in $\sOB\left((n+1)^{n (\omega+1)+2} N^{(\omega +2) n +2 } \right)$, where $\omega$ denotes the exponent in the complexity of matrix multiplication \cite{brand_complexity_2016}. However, such a method would not exploit the special structure of the system.

\vspace{0.15cm}
\vspace{-0.45cm}
\subsection*{Our approach and contribution}
We generalize the results of \cite{Diatta2022} for any $n>1$.
By following closely their techniques, we are able to provide amortized bounds on the separation of $F$. 
On solving the system of Eq.~(\ref{eq:introsystem}), the idea is to approximate the coefficients of $F$ for every $\x=(x_1, \ldots, x_n) \in \CC^n$ that is root of $F_1 = \ldots F_n=0$, up to a certain precision, so that to isolate the roots of $F(\x,Y)$, it suffices to isolate the roots of its approximation. 
The amortized bounds that we prove in Cor.~\ref{cor1} and Cor.~\ref{cor2}, quantify the required precision.
To find the roots of the approximation of $F(\x,Y)$ we can now use algorithms for univariate root isolation. Particularly, we employ the algorithm of \cite{MEHLHORN201534} that builds upon  the algorithm of approximate factorization of a polynomial of Pan \cite{PAN2002701}; if a univariate polynomial is of degree $d$, the approximate factorization algorithm returns $d$ root approximations.
Then, the approximations must be clustered in a way so that each cluster corresponds to a root, and it contains as many root approximations as the multiplicity of the corresponding root.
In \cite{MEHLHORN201534} they run Pan's algorithm multiple times with increasing precision.
For a stopping criterion, they require as input the number of distinct roots of the polynomial. Therefore, in our case, we  should also compute  the number of distinct roots of $F(\x,Y)$ for every root $\x \in \CC^n$ of $F_1 = \ldots = F_n=0$. This dominates the total bit-complexity.


In Sec.~\ref{sec:solving}, we compute the number of distinct roots of $F(\x,Y)$ using a numerical approach (Lem.~\ref{lem:numberofroots}). We compute the principal subresultant coefficients of $F$ and $\frac{\partial F}{\partial Y}$ with respect to $Y$. Then, for every root $\x$ of $F_1 = \ldots = F_n=0$, we approximate the principal subresultant coefficients up to the necessary precision so that we can determine their sign correctly. The index of the first non-zero subresultant coefficient gives the degree of the $\gcd$ of $F(\x,Y)$ and $\frac{\partial F(\x,Y)}{\partial Y}$, and thus the number of distinct roots. 
The total bit-complexity of solving the system of Eq.~(\ref{eq:introsystem}) then is described in Thm.~\ref{thm:complexity}(ii). In Rem.~\ref{rem:simplified}, we give for simplicity the bound for the case when all the polynomials have degree at most $d$ and bitsize $\tau$, which is 
$\sOB\left(n^2d^{3n+3}\tau +n^3 d^{2n+4}\tau \right) $.
On the contrary, when the number of distinct roots of $F(\x, Y)$ is known for every $\x$, or when $F$ is does not have multiple roots, we can isolate the roots of the system in 
$\sOB( nd^{2n+2}(d+n\tau))$ 
bit operations. 
This is to be compared with the result of \cite{StrETmult}; it improves it by a factor of $n$.

In Sec.\ref{sec:tsp} we apply the aggregate bounds of $F$ on the `Sum of Square Roots of Integers' Problem. Comparing the length of two paths in the Euclidean Travel Salesman Problem (TSP), also relates  to this problem. It has been already studied through the separation bound computation of the associated polynomial system \citep{Burnikel2000ASA,mehlhorn2000generalized}. Our approach matches the latter results, and, even more, the proven bounds are aggregate. 

\raggedbottom
\section{Notation and Prerequisites}
\label{sec:notation}

Let $n\in \NN$. We use the abbreviation $[n]$ for the set $\{1, \dots, n\}$.
We denote vectors by boldface symbols.
 For a vector $\x = (x_1,\dots,x_n) \in \CC^n$,  we denote the vector $(x_1,\dots, x_{i-1},x_{i+1},\dots, x_n) \in \CC^{n-1}$, $i\in [n]$ by $\x_{-i}$.
We call absolute $L$-bit approximation of a real number $a$, a rational number $\tilde{a}$ such that $|a-\tilde{a}|<2^{-L}$. 
We denote the arithmetic, resp.  bit,
complexity by $\OO$, resp. $\OB$ and we use $\sOO$, resp. $\sOB$, to ignore
(poly-) logarithmic factors.

For a polynomial  $f = \sum_{i=1}^d a_iX^i \in \CC[X]$ we denote its $\ell_1$-norm by $\norm{f}_1$, i.e., $\norm{f}_1 = \sum_{i=1}^d |a_i|$, its $\ell_2$-norm by $\norm{f}_2$, i.e., $\norm{f}_2 = \sqrt{\sum_{i=1}^d |a_i|^2}$ and its $\ell_\infty$-norm by $\norm{f}_\infty$, i.e., $\norm{f}_\infty = \max_{i\in \{0,\dots, d\}} |a_i|$.
We denote the leading coefficient of $f$ by $\lc(f)$. 
The $k$-th derivative of $f$ is denoted by $f^{(k)}$ and $f^{[k]}:= \frac{f^{(k)}}{k!}.$
When $f$ has integer coefficients, the \textit{bitsize} of the polynomial is defined as the logarithm of its $\ell_\infty$-norm. 
All the logarithms in the present paper are of base 2. 
A univariate polynomial with integer coefficients is of \textit{size} $(d,\tau)$ when its degree is at
most $d$ and it has bitsize $\tau$.  
Similarly, a multivariate polynomial with integer coefficients is of \textit{size} $(d,\tau)$ when its total degree is at
most $d$ and it has bitsize $\tau$.

Let $I = \langle f_1,\dots, f_k \rangle$ be a polynomial ideal in $\CC[X_1,\dots,X_n]$, $k \in \NN$. We denote the complex variety defined by $I$ by $V_\CC(I)$ or $V_\CC(f_1,\dots, f_k)$.
For a $\x\in V_\CC(I)$, we denote its multiplicity as root of $I$ by $\mult{I}(\x)$.  When $I$ is generated by one polynomial $f \in \CC[X]$, we write for simplicity $\mult{f}(x)$ to denote the multiplicity of $x$ as root of $I = \langle f\rangle$. The ideal $ \langle f_1,\dots,f_{i-1},f_{i+1}, \dots, f_k \rangle$, for $i\in \{1,\dots, k\}$ is denoted by $I \setminus f_i$.

\subsection{Univariate polynomials: some bounds and root isolation}

Let a univariate polynomial $f \in \CC[X]$ and $x \in V_\CC(f)$.
 The \textit{local separation of $f$ at $x$} is 
\begin{align*}
\sepa(x,f): = \min_{y\in V_\CC(f), y\neq x} |y-x|.
\end{align*}
The \textit{separation} of $f$ is
\begin{align*}
\sepa(f) := \min_{x\in V_\CC(f)} \sepa(x,f).
\end{align*}
The \textit{Mahler measure} of $f$ is defined as 
\begin{align*}
\mah(f) := |\lc(f)| \prod_{x\in V_\CC(f)} \max(1,|x|)^{\mult{f}(x)}.
\end{align*}
The following inequality bounds the Mahler measure of $f$ by means of its $\ell_1$ and $\ell_2$ norms \cite[Prop.10.8 and Prop.10.9]{BPR03}:
\begin{align}
\label{eq:bound_mah1}
2^{-d} \norm{f}_1 \le \mah(f) \le \norm{f}_2.
\end{align}
In particular, if $f \in \ZZ[X]$ and it is of size $(d,\tau)$, the previous inequality becomes
\begin{align}
\label{eq:bound_mah}
2^{-d} \norm{f}_1 \le \mah(f) \le \norm{f}_2\le 2^{\tau+ \log(d+1)}.
\end{align}
Following \cite[Def.~2.3 and Prop.~2.4]{Diatta2022}, we introduce the definition of the \textit{generalized discriminant} of $f\in \CC[X]$, which is 
\begin{align*}
\gdisc(f):= \lc(f)^{d-2} \prod_{x \in V_\CC(f)} f^{[\mult{f}(x)]}(x)^{\mult{f}(x)}\, .
\end{align*}
It plays an important role in the expression of  several bounds in the sequel.
We also define
\begin{align*}
	\lgdisc(f):= \sum_{x \in V_\CC(f)} \mult{f}(x)| \log (|f^{[\mult{f}(x)]}(x)|)|  \, \text{ and } \\
	\lsepa(f) := \sum_{x \in V_\CC(f)} \mult{f}(x)| \log (\sepa (x,f))|\, .  
\end{align*}
The next proposition, provides a bound for $\lsepa(f)$ by means of $\log \mah(f)$ and $\lgdisc(f)$.

\begin{proposition}[{\cite[Prop.~2.7]{Diatta2022}}] 
\label{prop:sep}
For a polynomial $f\in \CC[X]$ of degree $d$ with $| \lc(f)| \ge 1$, it holds that
$$\lsepa(f)\in \OO(d^2 +d \log \mah(f) +  \lgdisc(f)) .$$
\end{proposition}

To isolate the roots of a univariate polynomial with coefficients in $\CC$ we will use the algorithm of Mehlhorn et al. \cite{MEHLHORN201534}.  The algorithm requires that the number $k$ of distinct roots is known. It first computes an approximate factorization of the polynomial using Pan's algorithm \cite{PAN2002701} with an initial precision. Assuming that the polynomial is of degree $d$, from the approximate factorization one obtains approximations $\tilde{z}_1,\dots, \tilde{z}_d$ of the roots. Then,
the roots are grouped in $k$ clusters based on geometric vicinity. Every cluster is enclosed by a disc, each one corresponding to a root. If the discs are disjoint and each one contains the same number of root approximations as the multiplicity of the corresponding root, then the algorithm terminates. Otherwise, the factorization is repeated with increased precision. 
The next proposition summarizes their result.

\begin{proposition}[{\cite[Thm.~3]{MEHLHORN201534}, \cite[Prop.~2.22]{Diatta2022}}]
\label{prop:solv}
Let $f(x) \in \CC[x]$ of degree $d \ge 2$, for whom it holds that  $1/4 \le |\lc(f)| \le 1$. We assume that the number of distinct roots of $f$ is known. We can compute isolating discs for all $x\in V_\CC(f)$, as well as their multiplicities, in 
\begin{align*}
\sOB\big(d^3 + d^2 \log \mah(f) + d \lgdisc(f) \big)\, . 
\end{align*}
As input, we need an oracle giving an absolute $L$-bit approximation of the coefficients of $f$  with $L$ bounded by 
\begin{align*}
\sOO\big( d\log \mah(f) +\lsepa(f) +\lgdisc(f) \big) \, . 
\end{align*}
\end{proposition}

\subsection{Evaluation of polynomials}

If we want to evaluate a univariate polynomial $f \in \CC[X]$ of degree $d$ at some numbers $a_1,\dots,a_D\in \CC$, $D\in \NN$, we can use multipoint evaluation \cite{DBLP:journals/corr/abs-1304-8069}. When $D> d$, we have to repeat multipoint-evaluation $\left\lceil \frac{D}{d} \right \rceil$ times. When $D\le  d$ we have the following theorem:

\begin{theorem}[{\cite[Thm.9]{DBLP:journals/corr/abs-1304-8069}}]
\label{lem:multi1}
Let $f\in \CC[X]$ be a polynomial of degree $d$, with absolute value of coefficients at most $2^\tau$, and let $a_1, \dots, a_d\in \CC$ be complex points with absolute values bounded by $2^\Gamma$, where $\Gamma \ge 1$. Then, approximate multipoint evaluation up to a precision of $2^{-L}$ for some integer $L \ge 0$, that is, computing $\tilde{f}_i$ such that  $|\tilde{f}_i - f(a_i)| \le 2^{-L}$ for all $i$, can be done in $$\sOB\left(d(L+\tau+d\Gamma)\right)$$ bit-operations. The precision demand on $f$ and the points $a_i$ is bounded by $L+\sOO(\tau+d\Gamma)$ bits.
\end{theorem}

Now, we want to evaluate a multivariate polynomial $f\in \CC[\X]$ 
 at $\a_1,\dots, \a_D \in \CC^n$. 
As discussed in \cite{VANDERHOEVEN2020101405}, multipoint evaluation in the multivariate case is not an elementary extension of the univariate case, unless the evaluation points have good properties. In particular, when the evaluation points belong in a set of the form $S_1 \times \dots \times S_n$, it is advantageous to perform multipoint evaluation at each coordinate one by one. The advantage comes from the fact that the number of different values in each coordinate is $|S_i|$, whereas the evaluation points are in total $\prod_{i=1}^n |S_i|$.

\begin{proposition}
\label{prop:meval}
Let $f \in \CC[\X]$ be a polynomial of degree $d$ in each variable with absolute value of coefficients at most $2^\tau$ and $S = S_1 \times \dots \times S_n \subset \CC^n$ be  a set of complex points with absolute values bounded by $2^\Gamma$, where $\Gamma \ge 1$, and $|S_i| \le M$. Then, approximate multipoint evaluation up to a precision of $2^{-L}$ for some integer $L \ge 0$, that is, computing $\tilde{f}_i$ such that  $|\tilde{f}_i - f(\a_i)| \le 2^{-L}$ for all $i$, can be done in $$\sOB\left( \max(M, d)^{n-1}\left \lceil{\frac{M}{d}}\right \rceil(ndL +n^2 d \tau +n^3 \, d^2\, \Gamma  )\right)$$ bit-operations. The precision demand on $f$ and the points $a_j$ is bounded by $L + \sOO \left( n\tau + n^2d \Gamma\right)$ bits.
\end{proposition}

\begin{proof}
For any $k=1,\dots, n-1$, we can write $f$ as a polynomial in the variables $X_{k+1},\dots,X_n$ as $$f(\X) = \sum_{\substack{e_i \in \{0,\dots, d\}, \\i=k+1, \dots, n}} f_{e_{k+1}, \dots, e_n}\left(X_1,\dots,X_k\right) \, \,X_{k+1}^{e_{k+1}}\cdots X_n^{e_n}.$$
Then, $|f_{e_{k+1}, \dots, e_n}(X_1,\dots,X_k)|\le d^k \cdot 2^\tau \cdot 2^{kd\Gamma}$.
In particular, for $k=n-1$ and $\mathbf{a}_{-n} \in S_1\times \dots S_{n-1}$,  $f(\mathbf{a}_{-n},X_n) = \sum_{i=0}^d f_i(\mathbf{a}_{-n}) \, X_n^i$, with $|f_i(\mathbf{a}_{-n})| \le 2^{\sOO(\tau +(n-1)d\Gamma)}$. Evaluating $f(\mathbf{a}_{-n},X_n)$ at $S_n$ with precision $L$ can be done in $\sOB(d(L+\tau +(n-1)d\Gamma))$ with a precision demand on $f_i(\mathbf{a}_{-n})$ and on the points $a_n \in S_n$ in $L+\sOO(\tau+(n-1)\, d\, \Gamma)$ bits.

Recursively, for any $k\in \{1, \dots, n-1\}$ and $\mathbf{a}^{k-1}:= (a_1,\dots, a_{k-1}) \in S_1\times \dots \times S_{k-1}$, we need to evaluate the polynomials $f_{e_{k+1},\dots, e_n}(\mathbf{a}^{k-1},X_k)$ at $S_k$ with precision $L+ (n-k) \tau + d \Gamma \sum_{i=k}^{n-1} i $. 
Since the polynomial has coefficients with absolute value bounded by $\tau + (k-1)d\Gamma$, the required precision on the coefficients and on the points in $S_k$ is in $L+ \sOO((n-k+1) \tau +\sum_{i=k-1}^{n-1} i d \Gamma)$.
For a polynomial $f_{e_{k+1},\dots, e_n}$ this requires at most $M^{k-1} \cdot \left\lceil \frac{M}{d} \right\rceil$ multipoint evaluations of cost $\sOB(d(L +n\tau +n^2 d \Gamma))$ each one. For a fixed $k$ there are at most $(d+1)^{n-k}$ polynomials $f_{e_{k+1},\dots, e_n}$ to be evaluated, so this yields a complexity in $\sOB(d^{n-k} M^{k-1} \lceil \frac{M}{d} \rceil d(L +n\tau +n^2 d \Gamma))$ bit-operation.
By summing for all $k=1, \dots, n-1$ we obtain a total bit-complexity in 
\[
\sOB\left(\max(M,d)^{n-1}\left\lceil \frac{M}{d} \right\rceil nd(L +n\tau +n^2 d \Gamma)\right) \, ,
\]
to compute all the evaluations with an error bounded by $2^{-L}$ and a required precision of all the coordinates of the points in $S$ bounded by $L + \sOO(n\tau + n^2d \Gamma)$ bits.
\end{proof}

Notice that in both Thm.~\ref{lem:multi1} and Prop.~\ref{prop:meval}, the existence of an oracle providing the necessary approximations is assumed.

\section{Amortized Bounds for Polynomials in a Multiple Field Extension}
\label{sec:bounds}
Let $F_1\in \ZZ[X_1], \dots, F_n \in \ZZ[X_n]$ be univariate polynomials of size $(M,\Lambda)$ and $F \in \ZZ[\X,Y]$ of size $(d,\tau)$.
We consider the  ideals 
\begin{align}
	\label{eq:ideals-I-J}
	\mathcal{I} = \langle F_1, \dots, F_n \rangle \subset \ZZ[X_1,\dots, X_n]
	\quad \text{ and } \nonumber \\
	\mathcal{J} = \langle F_1, \dots,F_n ,F\rangle \subset \ZZ[X_1,\dots, X_n,Y].
\end{align}
For $\x\in V_\CC(\I)$, let $F_{\x}(Y) := F(\x,Y)$.
We prove aggregate separation bounds for the roots of $F$ in $(\ZZ[X_1,\dots,X_n])[Y]$.
We closely follow  {\cite{Diatta2022}}, where they treat the simple extension case, and we generalize their results to the $n$-variate field extension. 
We use Lem.~\ref{lem1mult} and Lem.~\ref{lem2mult}, which are generalizations of Prop.~2.10 and Prop.~3.3 of {\cite{Diatta2022}} respectively, as building blocks for our proofs.
Lem.~\ref{lem1mult} gives upper and lower bounds on the product of the evaluations of $n$-variate polynomials at all points in $V_\CC(\I)$ and in
Lem.~\ref{lem2mult}  the evaluation is of a set of $n+1$-variate polynomials at all points in $V_\CC(\J)$. 

First, due to the special structure of the ideals $\I$ and $\J$, we have the following result on the multiplicities of the roots of the corresponding varieties. It will be used in the proof of Lem.~\ref{lem1mult}.

\begin{lemma}[{\cite[Prop.3]{Imbach2021ClusteringCZ}, \cite{Zhang09}}]
\label{lem:mult} 
Let $\I$ and $\J$ be the ideals of  Eq.~\eqref{eq:ideals-I-J}.
For any $\x\in V_\CC(\I) $ and any $i\in [n]$ it holds that $\mult{\I}(\x) = \mult{F_i}(x_i) \cdot \mult{\I\setminus F_i}(\x_{-i})$. Moreover, for any $(\x,y) \in V_\CC(\J)$, it holds that $\mult{\J}(\x,y) = \mult{\I}(\x) \cdot \mult{F_{\x}}(y)$.
\end{lemma}

\begin{lemma}
\label{lem1mult}
Let $\I$ be the ideal of  Eq.~\eqref{eq:ideals-I-J}
and  $G_1,\dots , G_m \in \ZZ[X_1,\dots,X_n]$ of sizes $(\delta,\sigma)$. 

\noindent (i) Let $A\subseteq V_\CC(\mathcal{I})$ such that for every $\x \in A$, there exists an index $ i(\x)\in [m]$ such that $G_{i(\x)} (\x) \neq 0$. Then, 
\[
\sum_{\x \in A} \mult{\I}(\x) \log (|G_{i(\x)}(\x)| )\in \sOO\left(M^n(n+\sigma) +n\delta M^{n-1}\Lambda\right) \, . 
\]
(ii) If for every $\x  \in V_\CC(\mathcal{I})$ there exists an index $i\in [m]$ with $G_{i(\x)}(\x)\neq 0$, then we denote by $i(\x)$ the smallest such index. In this case,
\[
\sum_{\x \in V_\CC(\mathcal{I})} \mult{\I}(\x)\big| \log (|G_{i(\x)}(\x)|)\big| \in \sOO \left(M^n(n+\sigma) +n\delta M^{n-1}\Lambda \right) \, . 
\]

\end{lemma}

\begin{proof}
(i) For any $\x = (x_1,\dots,x_n)\in A$, 
\begin{align*}
	 |G_{i(\x)}(\x)| \le {{\delta+n}\choose{n}} 2^{\sigma} \prod_{j=1}^n \max\{1,|x_j|\}^{\delta},
\end{align*}
since the number of monomials in $\ZZ[X_1,\ldots, X_n]$ of degree less than or equal to $\delta$ is ${\delta+n}\choose{n}$, the absolute value of every coefficient is  $ \le 2^{\sigma}$ and every $x_j$ is of degree at most $\delta$.
Therefore, 
\begin{align}
\label{lemmult1:eq1}
	 \prod_{\x\in A}  |G_{i(\x)}(\x)|^{\mult{\I}(\x)} \le   \prod_{\x\in A}\Bigg( {{\delta+n}\choose{n}} 2^{\sigma} \prod_{j=1}^n\max\{1,|x_j|\}^{\delta}\Bigg)^{\mult{\I}(\x)}{\text{\hspace{-0.5cm}}}.
\end{align}
Since $\sum_{\x\in A} \mult{\I}(\x) \le M^n$ and ${{\delta+n}\choose{n}} \in \OO((\delta+n)^n)$, we have:
\begin{align}
\label{lemmult1:eq2}
\prod_{\x\in A} \Bigg({{\delta+n}\choose{n}} 2^{\sigma} \Bigg)^{\mult{\I}(\x)}   \in 2^{\sOO(M^n(n+\sigma))}.
\end{align}
For $j \in [n]$ it holds that
 \begin{align*}
	\prod_{\x\in A} \max\{1,|x_j|\}^{\delta \, \mult{\I}(\x)}   =  \prod_{\x\in A} \big(\max\{1,|x_j|\}^{\mult{F_j}(x_j)}\big)^{\delta \, \mult{\I \setminus F_j}(\x_{-j})}  \\ 
	 = \prod_{x_j\mid \x \in A} \big(\max\{1,|x_j|\}^{\mult{F_j}(x_j)}\big)^{\delta \, \sum_{\x_{-j} \mid \x \in A} \mult{\I \setminus F_j}(\x_{-j})} \\
	\le  \prod_{x_j\mid \x \in A} \big(\max\{1,|x_j|\}^{\mult{F_j}(x_j)}\big)^{\delta \, M^{n-1}} \le \mathcal{M}(F_j)^{\delta \,M^{n-1}},
	\end{align*}
where the first equality follows from Lem.~\ref{lem:mult} and the first inequality from the fact that $\sum_{\x_{-j} \mid \x \in A} \mult{\I \setminus F_j}(\x_{-j}) \le M^{n-1}$. Note that the last inequality is true since the coefficients of $F_j$ are in $\ZZ$ 
and so the absolute value of the leading coefficient of $F_j$ is greater or equal to 1.
We have that $\mathcal{M}(F_j) \in 2^{\OO(\Lambda +\log M)}$, following Eq.~(\ref{eq:bound_mah}). Therefore,
 \begin{align}
\label{lemmult1:eq3}
		 \prod_{\x\in A} \max\{1,|x_j|\}^{\delta \mult{\I}(\x)} \in 2^{\sOO(\delta M^{n-1}\Lambda)}.
	\end{align}
From the equations (\ref{lemmult1:eq1}), (\ref{lemmult1:eq2}) and (\ref{lemmult1:eq3}), we conclude.\\
(ii) Let $A = \{ \x \in V_\CC(\I) \mid |G_{i(\x)}(\x) | \ge 1\}$.
 Then, we can write:
\begin{align}
\label{lemmult2:eq1}
\sum_{\x \in V_\CC(\mathcal{I})} \mult{\I}(\x) \big| \log (|G_{i(\x)}(\x)| ) \big| = 2\sum_{\x \in A} \mult{\I}(\x) \log (|G_{i(\x)}(\x)| ) 
- \nonumber \\ - \sum_{\x \in V_\CC(\mathcal{I})} \mult{\I}(\x) \log (|G_{i(\x)}(\x)| ).
\end{align}
Using (i) of this lemma, we obtain an upper bound for the first term of the previous sum. Thus, we only need to compute a lower bound for $ \sum_{\x \in V_\CC(\mathcal{I})} \mult{\I}(\x) \log (|G_{i(\x)}(\x)| )$.
Let $G(X_1,\dots,X_n,U) = G_1(X_1,\dots, X_n) + G_2(X_1,\dots, X_n) U + \dots +G_m(X_1,\dots, X_n)U^{m-1}$. We consider $\res_{\X}(F_1,\dots,F_n,G)$, the multivariate resultant where we eliminate $\X$. Using the Poisson formula  \cite[Thm.~4.14]{BPR03} we can write
\begin{align*}
\res_{\X}(F_1,\dots,F_n,G)  =  \res_{\X}(\text{lc}(F_1)X_1^{\text{deg}(F_1)},\dots, \text{lc}(F_n)X_n^{\text{deg}(F_n)})^{\OO(\delta)} \cdot \nonumber \\ \cdot \prod_{\x \in V_\CC(\mathcal{I})}G(\x,U)^{\mult{\I}(\x)} = \\
= \Bigg(\res_{\X}(X_1^{\text{deg}(F_1)}, \dots,X_n^{\text{deg}(F_n)}) \prod_{j=1}^n \text{lc}(F_j)^{\prod_{k\in[n], k\neq j} \text{deg}(F_k)} \Bigg)^{\OO(\delta)}  \cdot  \\ \cdot \prod_{\x \in V_\CC(\mathcal{I})}G(\x,U)^{\mult{\I}(\x)} =  \\
= \Bigg(\prod_{j=1}^n \text{lc}(F_j)^{\prod_{k\in[n], k\neq j} \text{deg}(F_k)} \Bigg)^{\OO(\delta)}\prod_{\x\in V_\CC(\mathcal{I})}G(\x,U)^{\mult{\I}(\x)},
\end{align*}
which is a polynomial in $\ZZ[U]$; it is not identically zero, since by the hypothesis, for every $\x\in V_\CC(\I)$, $G(\x, U)$ is not identically zero. The absolute value of the constant term
 of $\res_{\X}(F_1,\dots,F_n,G)$ is
 \begin{align}
\label{lemmult2:eq2}
\Bigg|\prod_{j=1}^n \text{lc}(F_j)^{\prod_{k\in[n], k\neq j} \text{deg}(F_k)} \Bigg|^{\OO(\delta)}\prod_{\x \in V_\CC(\mathcal{I})}|G_{i(\x)}(\x)|^{\mult{\I}(\x)} \ge 1.
\end{align}
Since $\left|\prod_{j=1}^n \text{lc}(F_j)^{\prod_{k\in[n], k\neq j} \text{deg}(F_k)} \right| \in 2^{\OO(n \Lambda M^{n-1})}$, it follows from Eq.~(\ref{lemmult2:eq2}) that
\begin{align}
\label{lemmult2:eq3}
\prod_{\x \in V_\CC(\mathcal{I})}|G_{i(\x)}(\x)|^{\mult{\I}(\x)} \in 2^{-\OO(n\delta \Lambda M^{n-1})}.
\end{align}

So, by applying part (i) of the lemma and Eq.~(\ref{lemmult2:eq3}) to Eq.~(\ref{lemmult2:eq1}), we conclude.
\end{proof}
The following corollary, provides an amortized bound on the sum of the logarithms (bitsize) of the Mahler measures of the polynomials $F_{\x}(Y)$, for all $\x \in V_\CC(\I)$ (counting multiplicities).

\begin{corollary}[Amortized Mahler measure]
\label{cor1}
Let $\I$  be the ideal of  Eq.~\eqref{eq:ideals-I-J}.
Then, 
\begin{align*}
\sum_{\x \in V_\CC(\mathcal{I})} \mult{\I}(\x) \log \mathcal{M}(F_{\x}) \in \sOO\left(M^n(n+\tau+d)+nM^{n-1}d\Lambda\right) .
\end{align*}
\end{corollary}

\begin{proof}
We write $F(\X, Y) = f_{d}(\X)Y^{d} +\dots + f_{0}(\X)$.
For any $\x  \in V_\CC(\mathcal{I})$, 
following  Eq.~(\ref{eq:bound_mah}), 
 it  holds
\begin{align}
\label{cor1:eq1}
2^{-d}\norm{F_{\x}}_1 \le \mah(F_{\x}) \le \norm{F_{\x}}_2\, ,
\end{align}
since the degree of any
 $F_{\x}(Y)$ is $\le d$. Let
\begin{align*}
|f_{M(\x)}(\x) |:= \max_{j\in \{0,\ldots d\}} |f_j(\x)| \, , \\
|f_{m(\x)}(\x) |:= \min_{j\in \{0,\ldots d\} \mid f_j(\x) \neq 0} |f_j(\x)| \, . 
\end{align*}
Now, Eq.~(\ref{cor1:eq1}) gives
\begin{align*}
	2^{-d} \, |f_{m(\x)}(\x)|
	\le \mathcal{M}(F_{\x}) 
	\le \sqrt{d+1} \, |f_{M(\x)} (\x)|\, .
\end{align*}
If we consider for all $\x \in V_\CC(\I)$ (counting multiplicities),
then 
\begin{align}
\label{cor1:eq2}
2^{-dM^n}\prod_{\x \in V_\CC(\mathcal{I})}|f_{m(\x)}(\x)|^{\mult{\I}(\x)} \le \prod_{\x\in V_\CC(\mathcal{I})}\mah(F_{\x})^{\mult{\I}(\x)} 
 \le \nonumber\\ \le \sqrt{d+1}^{M^n} \prod_{\x \in V_\CC(\mathcal{I})}|f_{M(\x)} (\x)|^{\mult{\I}(\x)} \, .
\end{align}
We can bound the products on each side of the inequality in Eq.~(\ref{cor1:eq2}) by Lem.~\ref{lem1mult}. This concludes the proof.
\end{proof}

The following lemma is an analog of Lem.~\ref{lem1mult}, but in the case where we evaluate over $V_\CC(\J)$. 

\begin{lemma}
\label{lem2mult}
Let $\I$ and $\J$ be the ideals of  Eq.~\eqref{eq:ideals-I-J}
and $G_1,\dots , G_m \in \ZZ[X_1,\dots,X_n,Y]$ of sizes $(\delta, \sigma)$.\\
(i) Let $A\subseteq V_\CC(\mathcal{J})$ such that for every $(\x,y) \in A$, there exists an index $i(\x,y) \in [m]$ such that $G_{i(\x,y)} (\x,y) \neq 0$. Then, 
\begin{align*}
\sum_{(\x,y) \in A} \mult{\J}(\x,y) &\log |G_{i(\x,y)}(\x,y)| \\ & \in \sOO \left(M^n(d(n+\sigma) +\delta(n+\tau+d))  +  n\delta d M^{n-1}\Lambda \right).
\end{align*}
(ii) Supposing that for every $(\x,y) \in V_\CC(\mathcal{J})$ there exists an index $i\in [m]$ with $G_i(\x,y)\neq 0$, we denote by $i(\x,y)$ the smallest such index. Then,
\begin{align*}
\sum_{(\x,y) \in V_\CC(\mathcal{J})} \mu(\x,y)\big| & \log (|G_{i(\x,y)}(\x,y)|)\big| \\ & \in \sOO \left(M^n(d(n+\sigma) +\delta(n+\tau+d))  +  n\delta d M^{n-1}\Lambda \right) .
\end{align*}
\end{lemma}

\begin{proof}
(i) For any $(\x, y)\in A$, 
\begin{align*}
	 |G_{i(\x,y)}(\x,y)| \le {{\delta+n+1}\choose{n+1}} 2^{\sigma} \prod_{i=1}^n\max\{1,|x_i|\}^{\delta} \max\{1,|y|\}^{\delta}\, ,
\end{align*}
since the number of monomials in $\ZZ[X_1,\ldots, X_n,Y]$ of degree less than or equal to  $\delta$ is ${\delta+n+1}\choose{n+1}$.
We have that:
\begin{align}
\label{eq:lemmult2-1}
		 \prod_{(\x,y)\in A} \bigg({{\delta+n+1}\choose{n+1}} 2^{\sigma} \bigg)^{\mult{\J}(\x,y)} 
 \in 2^{\sOO(M^nd(n+\sigma))}\, ,
	\end{align}
since $\sum_{(\x,y)\in A} \mult{\J}(\x,y) \le M^nd $.
For $j=1,\dots,n$:
        \begin{align*}
          	\prod_{(\x,y)\in A} \max\{1,|x_j|\}^{\delta \mult{\J}(\x,y)}    \le  \prod_{(\x,y)\in V_\CC(\mathcal{J})} \max\{1,|x_j|\}^{\delta \mult{\J}(\x,y)}  \\
	 =  \prod_{(\x,y)\in V_\CC(\mathcal{J})} \max\{1,|x_j|\}^{\delta \mult{F_j}(x_j) \mult{\I_{-j}}(\x_{-j}) \mult{F_{\x}}(y)} \\
\le  \prod_{\x \in V_\CC(\mathcal{I})} \max\{1,|x_j|\}^{\delta \mult{F_j}(x_j) \mult{\I_{-j}}(\x_{-j}) d}  \\
 = \prod_{x_j\in V_\CC(F_j)} \max\{1,|x_j|\}^{\delta \mult{F_j}(x_j) \sum_{\x_{-j} \mid \x \in V_\CC(\I)} \mult{\I\setminus F_j}(\x_{-j}) d } \\ 
 \le   \prod_{x_j\in V_\CC(F_j)} \max\{1,|x_j|\}^{\delta  \mult{F_j}(x_j) M^{n-1}d}  
 \le \mah(F_j)^{\delta M^{n-1}d} ,
	\end{align*}
where the first equality follows from Lem.~\ref{lem:mult} and the third inequality from the fact that $\sum_{\x_{-j} \mid \x \in V_\CC(\I)} \mult{\I \setminus F_j}(\x_{-j}) \le M^{n-1}$. Note that the last inequality is true since the coefficients of $F_j$ are in $\ZZ$.
Since $\mah(F_j) \in 2^{\OO(\Lambda +\log M)}$, we have that
\begin{align}
\label{eq:lemmult2-2}
		 \prod_{(\x,y)\in A} \max\{1,|x_j|\}^{\delta \mult{\J}(\x,y)}  \in 2^{\sOO(\delta M^{n-1}d\Lambda )}.
\end{align}

Lastly, we have that
\begin{align*}
\prod_{\x\in V_\CC(\I)} |\lc(F_{\x}(Y)|^{\delta \mult{\I}(\x)} \cdot \prod_{(\x,y)\in A} \max\{1,|y|\}^{\delta \mult{\J}(\x,y)}  \nonumber   \le \\ \le   \prod_{\x\in V_\CC(\I)} |\lc(F_{\x}(Y)|^{\delta \mult{\I}(\x)} \left( \prod_{y\in V_\CC(F_{\x})} \max\{1,|y|\}^{\mult{F_{\x}}(y)} \right)^{\delta \mult{\I}(\x)} \le \nonumber\\ 
 \le \prod_{\x \in V_\CC(\mathcal{I})} \mathcal{M}(F_{\x})^{\delta \mult{\I}(\x)} \,  \in  2^{\sOO(M^n\delta(n+\tau+d)+M^{n-1} \delta n d \Lambda)} \, ,
\end{align*}


which follows from Cor.\ref{cor1}. Also, from Lem.~\ref{lem1mult} we can bound the size of the factor $\prod_{\x\in V_\CC(\I)} |\lc(F_{\x}(Y)|^{\delta \mult{\I}(\x)}$ on the left-hand side of  the previous equation, and thus, we have that 
\begin{align}
\label{eq:lemmult2-4}
 \prod_{(\x,y)\in A} \max\{1,|y|\}^{\delta \mult{\J}(\x,y)}  \in 2^{\sOO(M^n\delta(n+\tau+d)+M^{n-1} \delta n d \Lambda)}.
\end{align}

 By putting together Eq.~(\ref{eq:lemmult2-1}), Eq.~(\ref{eq:lemmult2-2}) and Eq.~(\ref{eq:lemmult2-4}) we can conclude.

(ii)  
As in the proof of Lem.~\ref{lem1mult}, by the first part of the lemma for $A = \{(\x, y) \mid |G_{i(\x, y)}(\x, y)| \ge 1\}$, we just need to find a lower bound for 
$
\sum_{(\x,y) \in V_\CC(\mathcal{J})} \mult{\I}(\x)\mult{F_{\x}}(y) \log (|G_{i(\x,y)}(\x,y)| )\, . 
$
Let $G(\X,Y,U) := G_1(\X,Y) + G_2(\X,Y) U + \dots +G_m(\X,Y)U^{m-1}$.
Let also $Q(\X,U):= \res_Y(G(\X,Y,U), F(\X,Y))$, be the resultant where we eliminate $Y$. 
Without loss of generality, we assume that the leading coefficient of $F(X_1,\dots, X_n, Y)$ when considered as a polynomial in $\ZZ[X_1,\dots, X_n][Y]$, is not canceled for any root of $\mathcal{I}$ (in the case where the leading coefficient is cancelled for some roots, $F$ is replaced by a polynomial of smaller degree). So, the resultant is not the zero polynomial.

We consider $\res_{\X}(Q, F_1,\dots,F_n)$, which is now the resultant where we eliminate $\X$. Using the Poisson formula  \cite[Thm.~4.14]{BPR03} we can write
\begin{align*}
\res_{\X}(Q(\X,U), F_1(X_1), \dots, F_n(X_n))   =  \\ =  \Bigg(\prod_{j=1}^n \text{lc}(F_j)^{\prod_{k\in[n], k\neq j} \text{deg}(F_k)} \Bigg)^{\OO(d\delta)}\prod_{\x\in V_\CC(\mathcal{I})}Q(\x,U)^{\mu(\x)} =  \\
= \Bigg(\prod_{j=1}^n \text{lc}(F_j)^{\prod_{k\in[n], k\neq j} \text{deg}(F_k)} \Bigg)^{\OO(d\delta)}\prod_{\x \in V_\CC(\mathcal{I})} f_d(\x)^{\OO(\delta)\mult{\I}(\x)}   \cdot \\ \cdot \prod_{y\mid (\x,y) \in V_\CC(\mathcal{J})} G(\x,y,U)^{\mult{\I}(\x)\mu(y)} \,.
\end{align*}

The absolute value of the constant term of $\res_{\X}(Q(\X,U), F_1(X_1),$ $ \dots, F_n(X_n))\in \ZZ[U]$ is:
 \begin{align*}
 \Bigg|\prod_{j=1}^n \text{lc}(F_j)^{\prod_{k\in[n], k\neq j} \text{deg}(F_k)} \Bigg|^{\OO(d\delta)} \prod_{\x \in V_\CC(\mathcal{I})} |f_d(\x)|^{\OO(\delta)\mult{\I}(\x)} \cdot \\ \cdot \prod_{(\x,y) \in V_\CC(\mathcal{J})} |G_{i(\x,y)}(\x,y)|^{\mult{\I}(\x)\mult{F_{\x}}(y)} \ge 1.
\end{align*}
We have that $\Bigg|\prod_{j=1}^n \text{lc}(F_j)^{\prod_{k\in[n], k\neq j} \text{deg}(F_k)} \Bigg| \in 2^{\OO(n \Lambda M^{n-1})}$.
From Lem.~\ref{lem1mult}~(i) it follows that
\begin{align*}
\prod_{\x \in V_\CC(\mathcal{I})}|f_d(\x)|^{\mult{\I}(\x)} \in 2^{\OO(M^n(n+\tau) + n d   M^{n-1}\Lambda)}\,  
\end{align*}
and thus 
\begin{align}
\label{eq:lowerbnd2}
\prod_{(\x,y) \in V_\CC(\mathcal{J})} |G_{i(\x,y)}(\x,y)|^{\mult{\I}(\x)\mult{F_{\x}}(y)}\in 2^{-\OO(\delta M^n(n+\tau) + n d \delta   M^{n-1}\Lambda)}\,  
\end{align}
So, by combining the first part of the lemma and Eq.~(\ref{eq:lowerbnd2}), we conclude.
\end{proof}

\begin{corollary}[Amortized  bound on $\lgdisc$ and $\lsepa$]
\label{cor2}
Let $\I$ and $\J$ be the ideals of  Eq.~\eqref{eq:ideals-I-J}. Then,
\[
\sum_{\x \in V_\CC(\mathcal{I}) } \mult{\I}(\x) \, \lgdisc(F_{\x}) \in \sOO\left( dM^n(n+\tau+d) +nd^2M^{n-1}\Lambda  \right)
\]
and
\[
\sum_{\x\in V_\CC(\mathcal{I}) } \mult{\I}(\x)  \, \lsepa(F_{\x}) \in \sOO\left( dM^n(n+\tau+d) +nd^2M^{n-1}\Lambda \right) \, . 
\]
\end{corollary}

\begin{proof}
We can write
\begin{align*}
\sum_{\x \in V_\CC(\mathcal{I}) } \mult{\I}(\x) \lgdisc(F_{\x})  = \\ =   \sum_{\x\in V_\CC(\mathcal{I}) } \mult{\I}(\x) \sum_{y \in V_\CC(F_{\x})} \mult{F_{\x}}(y)| \log (|F_{\x}^{[\mult{F_{\x}}(y)]}(y)|)|  =  \\ = \sum_{(\x,y)\in V_\CC(\mathcal{J}) } \mult{\J}(\x,y)| \log (|F_{\x}^{[\mult{F_{\x}}(y)]}(y)|)| 
\end{align*}
and then apply Lem.~\ref{lem2mult} for the family of polynomials $F_{\x}^{[k]}$, for $k=0,\dots, d$. These polynomials are of size $(d,\sOO(\tau))$, therefore the first part follows.
The second part is an immediate consequence of Prop.~\ref{prop:sep}, Cor.~\ref{cor1} and the first part of this corollary.
\end{proof}

\section{Solving in a Multiple Field Extension}
\label{sec:solving}

In this section, we study the complexity of isolating the roots of the system in Eq.~(\ref{eq:introsystem}).
We first solve the univariate polynomials of the system and then, for every $\x \in V_\CC(\I)$, we will isolate the roots of $F_{\x}$.
Following \cite{Diatta2022}, we employ the univariate root isolation algorithm of Prop.~\ref{prop:solv}. The main result of this section is summarized in the theorem that follows.

\begin{theorem}
\label{thm:complexity}
(i) If the number of distinct roots of $F_{\x}(Y)$ for every $\x\in V_\CC(\I)$ is known, then we compute isolating discs for all the roots and the corresponding multiplicities in 
{\small
\begin{align*}
\sOB \Big(  nM^{n+1}d(n+\tau+d)  +nM^{n}d^2(n\Lambda +d^n(n+\tau+d+\Lambda))  +n^2d^{n+3}M^{n-1}\Lambda   \Big)\, . 
\end{align*}}
(ii) If the number of distinct roots is not known,  then we compute isolating discs for all the roots, together with the corresponding multiplicities in 
{\small
\begin{align*}
\sOB \bigg(\max(M, d^2)^{n-1}\left \lceil{\frac{M}{d^2}}\right \rceil((nM^n+n^2)d^5(\tau+n)  +n^2 \, d^6\, \Lambda (M^{n-1}+n)  ) +
\\+nM^{n}d^{n+2}(n+\tau+d+\Lambda)+n^2d^{n+3}M^{n-1}\Lambda   \bigg)\, . 
\end{align*}}

\end{theorem}

\begin{remark}
\label{rem:simplified}
When $M=d$ and $\Lambda = \tau$, the bit-complexity bounds of the previous theorem become
\[ \sOB( nd^{2n+2}(d+n\tau)) \, , \]
when the number of distinct roots of $F_{\x}(Y)$ is known for every $\x \in V_\CC(\I)$ (or when $F$ is squarefree) and
$$\sOB\left(n^2d^{3n+3}\tau +n^3 d^{2n+4}\tau \right)\, , $$
otherwise.
\end{remark}
\vspace{-0.25cm}
\begin{remark}[Number of distinct roots, numerically vs formally]
Determining the number of distinct roots of every $F_{\x}$ in Thm.~\ref{thm:complexity} dominates the total complexity.
When $n=1$, a formal method involving univariate gcd computations can be used to find this number; the initial system is triangular and it can be efficiently decomposed into regular triangular systems, for whom the number of distinct roots over every $\x$ is constant \cite{Diatta2022,LAZARD2017123}. However, when $n>1$ the system is not triangular and decomposing it to a set of regular triangular systems (cf. extended gcd computation \cite{10.1145/1113439.1113457}) as for the case $n=1$, and thus loosing the original shape of the system, would require isolating roots of a triangular system in $n$ variables. The latter problem is substantially more demanding than isolating the roots of $n$ univariate polynomials. 
For this reason, we will follow a numerical approach to find the number of distinct roots of every $F_{\x}$.

\end{remark}

In the remark that follows, we compute the complexity of isolating the roots of the system of Eq.~(\ref{eq:introsystem}) using the algorithm of approximate factorization of Pan \cite{PAN2002701} with the maximal precision; instead of requiring the number of distinct roots of a univariate polynomial, if we approximate its roots with precision up to the separation bound, then the root approximations that are have pairwise distances smaller that the separation bound, will correspond to the same root. So, this is a method that avoids computing the number of distinct roots of the univariate polynomials $F_{\x}$. Nevertheless, it is a theoretical approach which brings about practical limitations in contrast to our adaptive method. 

\begin{remark}[Pan's algorithm with maximal precision]
\label{rem:pan}
On isolating the roots of $F_{\x}$ for every $V_\CC(\mathcal{I})$, instead of employing the algorithm of Prop.~\ref{prop:solv}, that requires knowing the number of distinct roots, we can use Pan's algorithm of approximate factorization with precision up to the separation bound of the roots of the initial system of Eq.~(\ref{eq:introsystem}).
Then, for every $\x \in V_\CC(\I)$, we approximate $F(\x,Y)$ up to $ \sum_{\x\in V_\CC(\mathcal{I})}\mult{\I}(\x)\, \lsepa(F_{\x}) \in \sOO(M^nd(n+\tau+d) +nd^2M^{n-1}\Lambda)  )$ bits. From Prop.~\ref{prop:meval} this is done in 
{\small
\begin{align*}
    \sOB\left( \max(M, d)^{n-1}\left \lceil{\frac{M}{d}}\right \rceil(nM^nd^2(n+\tau+d) +n^2d^3M^{n-1}\Lambda +n^2 d \tau +n^3 \, d^2 \tau  )\right)
\end{align*}}
since there are $d$ polynomials to evaluate.
Then, for every $F_{\x}(Y)$, Pan's algorithm runs in
\begin{align*}
    \sOB\left( M^{2n}d^2(n+\tau+d) +nd^3M^{2n-1}\Lambda\right)\, 
\end{align*}
and returns isolating intervals for all the roots. 
\end{remark}

Before presenting the proof of Thm.~\ref{thm:complexity}, we need some intermediate results.
The next lemma, gives upper and lower bounds on the evaluation of a polynomial over an algebraic number. For simplicity, we ignore the logarithmic factors. 
\begin{lemma}
\label{lem:boundsB}
For every $\x \in V_\CC(\I)$ and a polynomial $b(\X)\in \ZZ[\X]$ of size $(\delta, \sigma)$, it holds that
$$ 2^{-\sOO( M^n(\sigma +n) + nM^{n-1}\delta\Lambda)} \le |b(\x)| \le 2^{ \sOO( n\delta \Lambda + \sigma )}.$$
\end{lemma}

\begin{proof}
For the upper bound, we have that 
\begin{align*}
|b(\x)|  \le {{\delta + n}\choose{n}} 2^\sigma \prod_{i\in [n]} \max(1, |x_i|)^\delta  \le {{\delta + n}\choose{n}} 2^\sigma 2^{\delta \OO(\Lambda)  n },
\end{align*}
since, for $i\in [n]$, $\max(1, |x_i|) \le 2^{\OO(\Lambda)}$ from the Cauchy bound \cite[Lem.~10.2]{BPR03}.
For the lower bound, we follow the technique of $u$-resultant and consider the system:
\begin{align}
\label{eq:lbound}
F_0(\X,Y) = F_1(X_1) = \dots = F_n(X_n) = 0,
\end{align}
where $F_0(\X,Y) = Y - b(\X)$ and $Y$ is a new variable.  
We consider the resultant of the previous system that eliminates $\X$. 
Then, 
\begin{align*}
R(Y):=\res_{\X}(F_0, \dots, F_n) = \lc(R) \prod_{\a \in V_\CC(\I)} (Y - b(\a))^{\mult{\I}(\a)} \in \ZZ[Y].
\end{align*}

We will find an upper bound on the bitsize of $R$. To this scope, we follow the proof of the DMM bound in \cite{DMM-j-2020} (see also the proof of the sparse resultant's height bound in \cite{10.1145/3087604.3087653} ). 
For $i=0,\dots,n$, let $Q_i$ be the Newton polytope of $F_i$. We denote by $\# Q_i$ the number of lattice points in the closed polytope $Q_i$ and by $M_i$ the mixed volume of all these polytopes except from $Q_i$.
The resultant $R$ is a univariate polynomial in $Y$, with coefficients homogeneous polynomials in the coefficients of the polynomials of the system in Eq.~(\ref{eq:lbound}):
\begin{align*}
R(Y) = \ldots + \rho_k Y^k \mathbf{c}_{0,k}^{M_0-k} \mathbf{c}_{1,k}^{M_1} \cdots \mathbf{c}_{n,k}^{M_n} + \ldots,
\end{align*}
where $\rho_k \in \ZZ$, $\mathbf{c}_{i,k}^{M_i}$ is a monomial in the coefficients of $F_i$ with total degree $M_i$, for $i\in [n]$, and $\mathbf{c}_{0,k}^{M_0-k}$ is a monomial in the coefficients of $F_0$ of total degree $M_0-k$.
It holds that 
\begin{alignat*}{4}
| \mathbf{c}_{1,k}^{M_1} \cdots \mathbf{c}_{n,k}^{M_n} | &\le & \prod_{i=1}^n \norm{F_i}_\infty^{M_i} & \le  2^{nM^{n-1}\delta\Lambda}\, ,\\
|\mathbf{c}_{0,k}^{M_0-k}| &\le & \norm{F_0}_\infty^{M_0-k}& \le 2^{\sigma(M^n-k)}\, ,\\
\rho_k &\le & \prod_{i=0}^n (\# Q_i)^{M_i} & \le ((\delta+1)^n+1)^{M^{n}} \, (M+1)^{nM^{n-1}\delta}.
\end{alignat*}


So, 
$\norm{R}_\infty \le 2^{nM^{n-1}\delta(\Lambda + \log M +1))+ M^n(\sigma + n\log \delta +n+1)}\, . $
Since $R(Y)$ is a polynomial with integer coefficients, from the Cauchy bound \cite[Lem.~10.3]{BPR03}, we have that the absolute value of any of its roots is $ \ge |\mathrm{tc}(R)| \norm{R}_\infty ^{-1}$, where $\mathrm{tc}(R)$ is the tailing coefficient of $R$.
\end{proof}

\begin{lemma}
\label{lem:degree}
For all $\x \in V_\CC(\I)$, we compute the degree of $F_{\x}(Y)$ in  bit-complexity in
{\small $$
\sOB\left( \max(M, d)^{n-1}\left\lceil{\frac{M}{d}}\right\rceil(nd^2(M^n(\tau+n) + n\tau) + n^2d^3\Lambda (M^{n-1}+n) )\right).
$$}
\end{lemma}

\begin{proof}
To determine which is the first non-zero coefficient of $F_{\x}(Y) = \sum_{i=0}^d f_i(\x) Y^i$, it suffices to approximate its coefficients up to $L$ bits, where
 $L \in \sOO( M^n(\tau + n) + nM^{n-1}d\Lambda)$ (Lem.~\ref{lem:boundsB}). From Prop.~\ref{prop:meval}, this can be done, for all $f_i(X)$ and all $\x \in V_\CC(\I)$, using multipoint evaluation, in 
\begin{align}
\label{eq:compl_degree}
\sOB\left( \max(M, d)^{n-1}\left\lceil{\frac{M}{d}}\right\rceil(nd^2L +n^2 d^2 \tau +n^3 \, d^3\Lambda  )\right)
\end{align}
bit-operations.
This requires approximations of every $\x$ to bit-accuracy at most 
$ \sOO \left( L+ nd +n^2 d \Lambda \right)$, which is done 
for all $\x \in \mathcal{V}_\CC(\I)$ in $\sOB( M^3  + M^2 \Lambda +M L +n M d + n^2 d M \Lambda)$ \cite[Thm.5]{MEHLHORN201534}.
The total bit-complexity is dominated by the one in Eq.~(\ref{eq:compl_degree}). We substitute the upper bound for $L$ to  conclude.
\end{proof}

\begin{lemma}
\label{lem:numberofroots}
For all $\x \in V_\CC(\I)$, we compute the number of distinct complex roots of $F_{\x}(Y)$ in bit-complexity in 
{\small
$$
\sOB\left( \max(M, d^2)^{n-1}\left \lceil{\frac{M}{d^2}}\right \rceil (nd^5(\tau+n)(M^n+n) + n^2 d^6 \Lambda (M^{n-1}+n))\right)\,.
$$}
\end{lemma}

\begin{proof}
We define the polynomials $F_\ell(\X,Y):= \sum_{i=0}^\ell f_i(\X) Y^i$, for $\ell =0, \dots, d$, 
which are truncated versions of $F(\X,Y)$. 
Since the degree of $F_{\x}$ is not the same for every $\x \in V_\CC(\I)$, we have to repeat the following steps for every $\ell = 0, \dots, d$:

(1) We compute the principal subresultant coefficients of $F_{\ell}(\X, Y)$, $\frac{\partial{F_\ell}}{\partial{Y}} (\X,Y)$ with respect to $Y$. The $j$-th principal subresultant coefficient is a polynomial in $\ZZ[\X]$, denoted by $\sres_j(\X)$, of total degree $\OO\left(\ell(\ell-j)\right)$ and bitsize $\sOO\left((\tau+n)(\ell-j)\right)$ \cite[Prop.8.72]{BPR03}. 
The computation of all principal subresultant coefficients is done in $\sOB\left(\ell^{2n+2} \tau\right)$ \cite[Lem.4]{LAZARD2017123}.

(2) The index of the first non-zero $\sres_j(\x)$ gives the degree of the $\gcd(F_\ell(\x, Y), \frac{\partial F_\ell}{\partial Y} (\x,Y))$.
So, for every $\x \in V_\CC(\I)$, we approximate $\sres_j(\x)$, for $j=0, \dots , \ell$, up to $L$ bits, with $L \in \sOO(M^n\ell(\tau+n) +nM^{n-1}\ell^2 \Lambda)$ (Lem.~\ref{lem:boundsB}),  so as to determine if it zero or not. 
From Prop.~\ref{prop:meval}, this can be done, for all $j \in\{ 0,\dots, \ell\} $ and all $\x \in V_\CC(\I)$, using multi-point evaluation, in 
{\small
\begin{align*}
    \sOB\left( \max(M, \ell^2)^{n-1}\left \lceil{\frac{M}{\ell^2}}\right \rceil (n\ell^4(\tau+n)(M^n+n) + n^2 \ell^5 \Lambda (M^{n-1}+n))\right)\,.
\end{align*}}

Repeating the previous steps for all $\ell \in \{0,\dots, d\}$, yields a total bit-complexity in
{\small
\begin{align}\label{eq:bitcompl}
 \sOB\left( \max(M, d^2)^{n-1}\left \lceil{\frac{M}{d^2}}\right \rceil (nd^5(\tau+n)(M^n+n) + n^2 d^6 \Lambda (M^{n-1}+n))\right)\,.
\end{align}}

As we can see in the proof of Prop.~\ref{prop:meval}, to compute these evaluations, we need to approximate each $\x$ to bit accuracy 
$$ \sOO \left( M^nd(\tau+n)+n^2d^2 M\Lambda + nd (\tau+n)\right).$$

 This costs for all $\x \in \mathcal{V}_\CC(\I)$, $\sOB(n\, M \, ( M^nd(\tau+n)+n^2d^2 M\Lambda + nd (\tau+n)))$ \cite[Thm.5]{MEHLHORN201534}, which does not overcome the bit-complexity of Eq.~(\ref{eq:bitcompl}).
\end{proof}

\begin{lemma}
\label{lem:approx}
For every $\x \in V_\CC(\I)$ let $\rho_{\x}$ be a positive integer. We compute $\rho_{\x}$-approximations of $F_{\x}(Y)$ for all  $\x \in V_\CC(\I)$ in 
{\small
\begin{align*}
\sOB \left( n \left( M^3 +M^2 \Lambda + M\max_{\x\in V_\CC(\I)} \rho_{\x} \right)  +  d^{n+1} \left( M^n(\tau+ d \Lambda)   + \sum_{\x\in V_\CC(\I)} \rho_{\x} \right)   \right) \, . 
\end{align*}}
\end{lemma}

\begin{proof}
For a $\rho_{\x}$-approximation of $F_{\x}(Y)$, it suffices to consider an $L_{\x}$-approximation of $\x$, where $   L_{\x} \in \sOO(\rho_{\x} +\tau +d \Lambda)$. This follows from \cite[Lem.1]{Becker2017CountingSO}, since the coefficients of $F(\x,Y)$ are polynomials in $ \ZZ[X_1,\dots, X_n]$ of size $(d,\tau)$, and $\max(1, \norm{\x}_\infty) \in 2^{\OO(\Lambda)}$.
To get the desired approximation of each $\x \in V_\CC(\I)$, we compute isolating discs of the roots of each $F_i$, for $i=1,\dots,n$, of size less than $2^{-L_{\max}}$, where $L_{\max} = \max_{x\in V_\CC(\I)} L_{\x}$.
This costs \cite[Thm.5]{MEHLHORN201534}
\begin{align}
\label{lemapprox:eq1}
\sOB\left(n (M^3 +M^2\Lambda + ML_{\max})\right).
\end{align} 

For a given $\x \in V_\CC(\I)$, to compute the $\rho_{\x}$-approximation of $F(\x,Y)$, we evaluate its coefficients at the $L_{\x}$-approximation of $\x$. This costs $\sOB(nd^{n+1}(\tau + L_{\x}))$: 
When we regard $F(\X, Y)$ as a polynomial in $Y$, it has at most $d+1$ coefficients which are polynomials in $\ZZ[\X]$ of size $(d,\tau)$. 
Using \cite[Lem.6]{Bouz14}, we evaluate each one of them in $\sOB(d^n(\tau + L_{\x}))$, and then we multiply the latter bound by $d$.
To obtain the final cost, we sum the latter bound for all $\x\in V_\CC(\I)$ and add the cost in Eq.~(\ref{lemapprox:eq1}) and use the fact that $L_{\max} \in \sOO(\max_{\x\in V_\CC(\I)} \rho_{\x} + \tau +d \Lambda)$.
\end{proof}

Now, by putting everything together, we can prove our main theorem.
\vspace{-0.1cm}
\begin{proof}[Proof of Thm.\ref{thm:complexity}]
(i)
For any $\bm{x} = (x_1, \dots, x_n) \in V_\CC(\mathcal{I})$,  we compute a $\tau_{\x}$ such that $2^{-\tau_{\x}-2} \le \lc(F_{\x}) \le 2^{-\tau_{\x}}$. Then, the polynomial $\tilde{F}_{\x} (Y) :=2^{-\tau_{\x}} F_{\x}(Y)$ satisfies the conditions of Prop.~\ref{prop:solv}, which we then use to isolate its roots (it has the same roots as $F_{\x}(Y)$). From Prop.~\ref{prop:solv}, we can solve every $\tilde{F}_{\x}(Y)$ in 
\begin{align}
\label{eq:compl}
\sOB(d( d^2 + d\log \mah(\tilde{F}_{\x}) + \lgdisc(\tilde{F}_{\x}) ).
\end{align}

For that, we require an approximation of the coefficients of  $\tilde{F}_{\x}( Y)$ to an absolute precision bounded by 
\begin{align*}
    \rho_{\x} \in \sOO(d\log \mah(\tilde{F}_{\x}) + \lsepa(\tilde{F}_{\x})+ \lgdisc(\tilde{F}_{\x}) ).
\end{align*}

Repeating the arguments as in the proof of \cite[Prop.~3.13]{Diatta2022} we have that $\sum_{\x \in V_\CC(\I)} \tau_{\x} \in \OO(\sum_{\x \in V_\CC(\I)}  \rho_{\x})$ and that the computation of $ \tau_{\x}$ does not affect the total complexity. 
Then, using Lem.~\ref{lem:approx}, we compute $\rho_{\x}$-approximations of the coefficients of $\tilde{F}_{\x}$ for all $\x \in V_\CC(\I)$ in 
\begin{align}
\label{eq:compl1}
\sOB \big( n (M^3 +M\Lambda(M+d) + M\max_{\x\in V_\CC(\I)} \rho_{\x} )  +  \nonumber \\+ nd^{n+1}M^n(\tau+ d \Lambda)  + nd^{n+1}\sum_{\x\in V_\CC(\I)} \rho_{\x}   \big)\, . 
\end{align}

From Cor.~\ref{cor1} and Cor.~\ref{cor2}, we have that 
\begin{align*}
    \sum_{\x\in V_\CC(\mathcal{I})}\mult{\I}(\x)\rho_{\x} \in \sOO(M^nd(n+\tau+d) +nd^2M^{n-1}\Lambda)  )\, . 
\end{align*}

Therefore,  Eq.~(\ref{eq:compl1}) becomes
\begin{align}
\label{eq:compl2}
\sOB \big( n (M^3 +M^2\Lambda) + nM^{n+1}d(n+\tau+d) + \nonumber \\ +nM^{n}d^2(n\Lambda +d^n(n+\tau+d+\Lambda)) +n^2d^{n+3}M^{n-1}\Lambda   \big)\, . 
\end{align}

By summing Eq.~(\ref{eq:compl}) for all $\x \in V_\CC(\I)$ yields
\begin{align}
\label{eq:compl3}
\sOB \big(M^nd^2(n+\tau+d) +nd^3M^{n-1}\Lambda  \big)\, . 
\end{align}

We add the bounds in Eq.~(\ref{eq:compl2}) and Eq.~(\ref{eq:compl3}) to conclude (the bound in Eq.~(\ref{eq:compl2}) dominates).

(ii) When the number of distict roots of $F_{\x}(Y)$ for every $\x \in V_\CC(\I)$ is not already known, then we have to compute it using Lem.~\ref{lem:numberofroots}. We add the cost of this computation to the bound of part (i) of the theorem.
\end{proof}

\section{Application: Sum of Square Roots of Integers Problem}
\label{sec:tsp}
We consider the problem of determining the minimum non-zero difference between two sums of square roots of integers. 
It appears as Problem 33 in `The Open Problems Project' and was originally addressed by Joseph O'Rouke \citep{Demaine2007TheOP}.

Let $a_i,b_i\in \ZZ^{\ge 0}$ for $i = 1,\dots, n$ of bitsize $\tau$.  We want to decide if $\sum_{i=1}^n\sqrt{a_i} $ is less than, equal to, or greater than $\sum_{i=1}^n\sqrt{b_i} $.
This problem is also related to the Euclidean Travel Salesman Problem (TSP): Given a set of points in the plane with integer coordinates and $L\in \NN$, decide if there exists a circuit passing through all these points and having total length (with respect to the Euclidean distance) at most $L$.
The length of the path is a sum of square-roots of integers. 

Comparing $\sum_{i=1}^n\sqrt{a_i}$ with $\sum_{i=1}^n\sqrt{b_i} $ in the real-RAM model, can be done trivially. However, in the bit-complexity setting, one has to determine the number of bits that is sufficient to obtain a correct result. 
We by $r(n,\tau)$ denote the minimum positive value of 
$
\big|\sum_{i=1}^n\sqrt{a_i} - \sum_{i=1}^n\sqrt{b_i} \big| \, . 
$
Lower bounds on $r(n,\tau)$, and in turn upper bounds on $-\log r(n,\tau)$, give upper bounds on the precision needed to compare $\sum_{i=1}^n\sqrt{a_i}$ with $\sum_{i=1}^n\sqrt{b_i} $. In particular, if $ -\log r(n,\tau)$ is bounded above by a polynomial in $k$ and $n$, then the sign of $\sum_{i=1}^n\sqrt{a_i} - \sum_{i=1}^n\sqrt{b_i}$ can be computed in polynomial time. 
Nevertheless, existing upper bounds on $-\log r(n,\tau)$ are exponential. In \citep{Burnikel2000ASA,mehlhorn2000generalized} they prove that $-\log r (n,\tau)\in \sOO(\tau 2^{2n})$, through studying separation bounds.

Here, we apply the results of Sec.~\ref{sec:bounds} to  derive bounds that, however, remain exponential in $n$. Nonetheless, the same bounds apply to the sum of all the roots of the associated system that has as root the two quantities that we have to compare. 
We consider the system
\begin{align*}
\begin{Bmatrix}
F_i(X_i):= X_i^2 - a_i= 0 \, , i\in [n]\\
G_i(Y_i):=Y_i^2 - b_i = 0 \, , i\in [n]\\
H(\X, \Y, Z):=(Z - X_1 - \dots - X_n ) (Z - Y_1 - \dots - Y_n ) = 0
\end{Bmatrix}
\end{align*} 

Let $\mathcal{K} = \langle F_1, \ldots, F_n, G_1, \ldots, G_n \rangle $. From Cor.\ref{cor2} we have that  
\begin{align*}
\sum_{(\x,\y)\in V_\CC(\mathcal{K}) } \mult{\mathcal{K}}(\x,\y) \lsepa (H(\x, \y, \cdot)) \in \sOO(n2^{2n}\tau )\, ,
\end{align*}
or equivalently, since all the multiplicities are equal to one, 
\begin{align}
\label{eq:tsp}
\sum_{(\x,\y)\in V_\CC(\mathcal{K}) } \big| \log \big| \sum_{i=1}^n x_i - \sum_{i=1}^n y_i \big| \, \big| \in \sOO(n2^{2n}\tau )\, .
\end{align}
We see that, as Eq.~(\ref{eq:tsp}) shows, not only $-\log r (n,\tau)$ is in $ \sOO(n
\tau2^{2n})$, but also the sum of the differences associated to all the $2^{2n}$ roots of the system $\{F_1=\dots=F_n=G_1=\dots=G_n= 0\}$.

\section{Concluding remarks}

We provided amortized bounds on the separation of a polynomial with coefficients in a multiple field extension. We used these bounds to 
estimate the bit-complexity of isolating its roots and applied them
to the `sum of square roots of integers' problem. For the root isolation, we followed an adaptive approach that we juxtaposed to a theoretical one that uses maximal precision, but leads to better bit-complexity estimates. In a future work, we set our sights on developing an adaptive method that matches the bit-complexity of the theoretical one. 

\paragraph{\textbf{Acknowledgements}} The authors would like to thank Elias Tsigaridas for the several discussions on the subject and Michael Sagraloff for his feedback and the suggestion of the theoretical approach.
\newpage

\bibliographystyle{ACM-Reference-Format}
\balance
\bibliography{msolve}

\end{document}